\documentclass[journal]{IEEEtran}
\usepackage{graphicx}
\usepackage{booktabs}
\usepackage{cite}
\usepackage{capt-of}
\usepackage{algorithm}
\usepackage{algpseudocode}

\usepackage{amsmath,amssymb,amsfonts,bm}
\usepackage{subcaption}
\usepackage{booktabs}
\usepackage{textcomp}
\usepackage{xcolor}
\usepackage{multirow}

\usepackage{tikz}
\usetikzlibrary{decorations.pathreplacing}   

\newtheorem{theorem}{Theorem}
\newtheorem{assumption}{Assumption}
\newtheorem{remark}{Remark}
\newtheorem{lemma}{Lemma}

\long\def\comment#1{}

\newfont{\bbb}{msbm10 scaled 700}

\newfont{\bb}{msbm10 scaled 1100}

\renewcommand{\det}{{\hbox{det}}}

\renewcommand{\arg}{{\hbox{arg}}}

\begin{document}

\title{CSI-Free Positioning of Movable Antennas for IoT Networks: A Compositional Kernelized Bandit}

\author{Wonseok Choi,~\IEEEmembership{Graduate Student Member,~IEEE,} Yunseob Tae,~\IEEEmembership{Graduate Student Member,~IEEE,} \\Jeongjae Lee,~\IEEEmembership{Member,~IEEE,} and~Songnam~Hong,~\IEEEmembership{Senior Member,~IEEE}
        \thanks{W. Choi, Y. Tae, and S. Hong are with the Department of Electronic Engineering, Hanyang University, Seoul, Korea. (e-mail: \{ryan4975, tys7524, snhong\}@hanyang.ac.kr).}

        \thanks{J. Lee is with the Department of Electronic Engineering, Hanyang University, Seoul, Korea, and the Ming Hsieh Department of Electrical Engineering, University of Southern California, Los Angeles, CA, USA. (e-mail: jl\_939@usc.edu).}
        
}

\maketitle

\begin{abstract}
Movable antenna (MA) arrays reshape the propagation channel by mechanically changing the element positions, which suits Internet-of-Things access points whose devices cannot adapt on their own. Existing position optimization assumes the channel is known at every candidate configuration. In practice a configuration can be evaluated only after the array has moved there, over several slots limited by the actuator speed, during which the channel changes. We formulate the positioning of $M$ MAs serving $K<M$ devices without channel state information as a nonstationary kernelized bandit over the joint configuration space, driven by a single scalar rate feedback per slot under a reachability constraint and an actuation-energy cost. Since the information gain of a standard kernel grows exponentially with $M$, and the multiuser sum rate depends on the positions through a sum of per-antenna terms, we design a compositional kernel that retains interactions up to antenna pairs. On this kernel we propose CoMoveUCB, which selects a target configuration over the whole feasible set by coordinate ascent and retains it under a persistence rule. We prove sublinear dynamic regret with a polynomial dependence on $M$. Simulations show that CoMoveUCB outperforms fixed, measure-then-optimize, and reachability-confined benchmarks across loadings, confirming the value of repositioning an MA array under realistic actuation limits.
\end{abstract}


\section{Introduction}
\label{sec:intro}

In a conventional antenna array the element positions are fixed. The channel realized at those positions is determined by the propagation environment alone and cannot be modified by the transceiver. Movable antenna (MA) systems~\cite{zhu2024movable} treat the positions as design variables, repositioning the elements mechanically within a confined region. Since the fading pattern decorrelates over distances on the order of a wavelength, a displacement of a few centimeters is sufficient to replace a deep fade with a favorable realization, and a larger displacement can move an element into or out of a blockage. In a multiuser setting the positions additionally determine the correlation between the devices' channel vectors, and therefore the degree to which a precoder can separate them. Position optimization of this kind has been applied to beamforming and rate maximization~\cite{ma2024multibeam,feng2024weighted}. MA arrays are particularly suited to an Internet-of-Things (IoT) access point (AP), whose devices are typically too simple and too power-constrained to adapt on their own: repositioning the antennas at the AP improves the link without requiring any action at the devices.

Mechanical actuation distinguishes an MA from a fluid antenna, which reconfigures by switching electronically among predefined ports~\cite{wong2021fluid}. Switching is instantaneous, whereas repositioning proceeds at the speed of the actuator. An antenna therefore crosses the region over many slots, transmits from intermediate positions rather than from its destination, and consumes energy in the motor-driven mechanism along the way~\cite{ning2025movable}. As carrier frequencies increase and coherence times shorten, this overhead occupies a non-negligible fraction of the transmission time. In an IoT deployment, where the AP may run on a battery or on harvested energy, the mechanical energy is itself a first-order cost. Several recent formulations account for the movement cost within the optimization: minimizing the repositioning delay for prescribed target positions~\cite{li2026trajectory}, maximizing throughput or energy efficiency under movement-delay penalties~\cite{wang2026throughput,ding2026energy}, and, in our prior work, movement-delay-aware sum-rate maximization for a linear MA array~\cite{tae2026delay}. These formulations, in common with most of the position-optimization literature, assume the channel is perfectly known in advance, which allows the reward of any candidate configuration to be evaluated without visiting it.

This premise is difficult to sustain in practice, since the reward at a new configuration can be measured only after the antennas have been relocated there. Acquiring channel state information (CSI) across the spatial domain through pilots incurs prohibitive overhead, which grows with the number of configurations to be probed and, in an IoT network, is paid by the devices in battery life. This has motivated CSI-free position optimization operating directly on received-signal measurements~\cite{zeng2025csi}, treating the configuration as the variable of a static optimization and refining it by zeroth-order gradient ascent. In the present setting, however, the reward drifts while the array is in transit, so a measurement is already outdated when the next one is collected, and each measurement requires a relocation limited by the actuator speed, at the cost of transmission time. Position selection is therefore a sequential decision problem rather than a one-shot optimization. For a \emph{single} MA, our prior work~\cite{choi2026positioning} formulated it as a nonstationary kernelized bandit under a reachability constraint and obtained sublinear dynamic regret.

An IoT AP typically serves multiple devices simultaneously, and spatial multiplexing of these devices requires an array. Extending the design of~\cite{choi2026positioning} to $M$ antennas is not a matter of running $M$ independent copies of it: the reward, the estimator, and the target selection are each affected differently, and the three resulting obstacles determine the structure of the proposed algorithm. First, the reward does not decompose across the antennas, since the signal-to-interference-plus-noise ratio~(SINR) of each device is governed by the correlation between the devices' channel vectors, a property of the array as a whole. Independent per-antenna learners would both discard this coupling and render each learner's environment nonstationary through the motion of the others, a nonstationarity generated by the algorithm itself that no assumption on the channel can bound. Second, a single scalar reward is observed per slot for an $M$-dimensional action, and nonstationarity precludes accumulating these observations: the estimator operates on a sliding window whose sample budget is independent of $M$, while the space to be covered is not. A generic kernel on the joint configuration space would further incur an information gain that grows exponentially with $M$. Third, the per-slot target is a configuration rather than a single position, and the maximization of the acquisition function, on which the regret guarantee rests, becomes a combinatorial problem over the joint configuration space that must be solved in every slot, with its approximation error entering the regret directly. To the best of our knowledge, CSI-free positioning of multiple movable antennas under these constraints has not been addressed.

The main contributions of this paper are as follows.
\begin{itemize}

\item We formulate CSI-free positioning of a movable antenna array serving multiple IoT devices as an online learning problem over the joint configuration of the array, in which the antennas are coupled through the multiuser sum rate and cannot be learned independently.

\item We design a compositional kernel that retains interactions up to antenna pairs, motivated by the additive per-antenna structure of the multiuser sum rate. It keeps the complexity of learning polynomial in the number of antennas, rather than exponential as under a generic kernel.

\item We propose CoMoveUCB, which combines the compositional kernel with a persistent target-and-approach scheme: a target is selected over the combinatorially large configuration space of the array, retained under a persistence rule, and approached within the speed and energy limits of mechanical actuation. Under standard assumptions, CoMoveUCB achieves sublinear dynamic regret.

\item Establishing this guarantee requires two results beyond the single-antenna analysis of~\cite{choi2026positioning}: a subadditivity bound on the information gain, which keeps the regret polynomial in the array size, and a travel bound for antennas moving toward a common target.

\item We evaluate CoMoveUCB against fixed, random, measure-then-optimize, and reachability-confined learning policies, all charged the same movement cost. CoMoveUCB outperforms every benchmark, and an ablation attributes most of the gain to its optimistic target selection and the proposed compositional kernel.

\end{itemize}

The remainder of this paper is organized as follows. Section~\ref{sec:system} describes the system model and formulates the positioning task as an online learning problem measured by dynamic regret. Section~\ref{sec:algorithm} presents the compositional kernel and the CoMoveUCB algorithm. Section~\ref{sec:analysis} establishes the regret guarantee, and Section~\ref{sec:sim} evaluates CoMoveUCB against the benchmarks. Section~\ref{sec:conclusion} concludes the paper.

\emph{Notation:} Boldface lowercase and uppercase letters denote vectors and matrices, respectively. $(\cdot)^{\mathsf{T}}$ and $(\cdot)^{\mathsf{H}}$ denote the transpose and the conjugate transpose, $\|\cdot\|_{1}$ and $\|\cdot\|_{\infty}$ the $\ell_{1}$ and $\ell_{\infty}$ norms, $\lambda_{i}(\cdot)$ and $\lambda_{\max}(\cdot)$ the $i$th and the largest eigenvalue, $\mathbf{A}\succeq\mathbf{B}$ that $\mathbf{A}-\mathbf{B}$ is positive semidefinite, and $\mathbb{E}[\cdot]$ the expectation. $O(\cdot)$ and $\Theta(\cdot)$ denote asymptotic upper and tight bounds, and $\tilde{O}(\cdot)$ hides polylogarithmic factors.

\begin{figure}[t]
\centering
\includegraphics[width=\columnwidth]{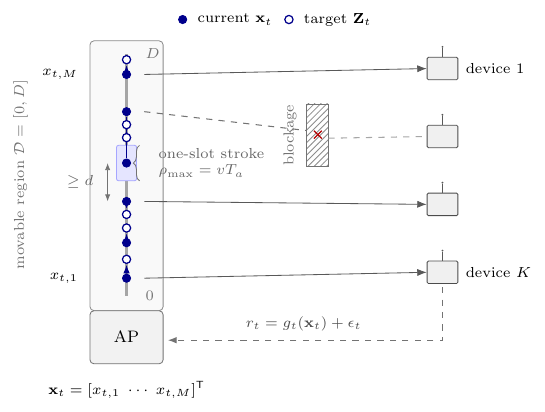}
\caption{Downlink MA system. An AP with $M$ movable antennas on a linear region $\mathcal{D}=[0,D]$, subject to a minimum spacing $d$, serves $K<M$ IoT devices. The array advances at most one stroke $\rho_{\max}$ per slot toward the target $\mathbf{Z}_{t}$, transmitting from intermediate APVs on the way, and observes a single noisy scalar $r_{t}$ per slot.}
\label{fig:system}
\end{figure}

\section{System Model and Problem Formulation}\label{sec:system}

We describe the MA array and its actuation constraints, and formulate the positioning task as a sequential decision problem measured by dynamic regret.

\subsection{System Model}
\label{subsec:sysmodel}

We consider a downlink scenario, illustrated in Fig.~\ref{fig:system}, in which an AP equipped with $M$ MAs serves $K$ single-antenna IoT devices, with $M>K$. Such loading is typical of IoT deployments, where a moderate number of low-complexity devices are served by a single infrastructure node. The MAs are deployed along a linear region $\mathcal{D}\triangleq[0,D]$, and time is divided into slots of duration $T_{\mathrm{s}}$, indexed by $t\in\{1,2,\dots\}$. The antenna position vector (APV) in slot $t$ is denoted by
$\mathbf{x}_{t}=[x_{t,1}\ \cdots\ x_{t,M}]^{\mathsf{T}}$, with a given initial APV $\mathbf{x}_{0}$. A minimum spacing $d>0$ is imposed to limit mutual coupling, and every APV therefore lies in
\begin{equation}
  \mathcal{C}\triangleq\big\{\mathbf{x}\in\mathcal{D}^{M}:
  x_{m+1}-x_{m}\ge d,\ m=1,\dots,M-1\big\},
  \label{eq:array-set}
\end{equation}
which also fixes the ordering of the antennas and hence their consistent indexing across slots.

Each slot comprises an actuation interval of duration $T_{\mathrm{a}}$, during which the antennas are repositioned by mechanical actuators of maximum speed $v$, followed by a transmission interval of duration $T_{\mathrm{s}}-T_{\mathrm{a}}$. Since the $M$ actuators operate in parallel, the actuation interval is determined by the \emph{largest} displacement among the antennas, which the stroke $\rho_{\max}\triangleq vT_{\mathrm{a}}$ bounds:
\begin{equation}
  \mathbf{x}_{t}\in\mathcal{A}(\mathbf{x}_{t-1})\triangleq\big\{\mathbf{x}
  \in\mathcal{C}:\|\mathbf{x}-\mathbf{x}_{t-1}\|_{\infty}\le\rho_{\max}\big\}.
  \label{eq:reachable}
\end{equation}
The actuation interval is provisioned for $\rho_{\max}$ and held fixed across slots. A policy may select a destination APV $\mathbf{Z}\in\mathcal{C}$ outside the reachable set. Reaching it then takes $\lceil\|\mathbf{Z}-\mathbf{x}_{t-1}\|_{\infty}/\rho_{\max}\rceil$ slots, on each of which the array is \emph{in transit} and transmits from an intermediate APV rather than from $\mathbf{Z}$.

Unlike the actuation delay, the mechanical energy is drawn by every actuator and accumulates over all displacements. The two costs are therefore governed by different norms: the delay by $\|\cdot\|_{\infty}$ as in~\eqref{eq:reachable}, and the energy of slot $t$ by $c_{\mathrm{e}}\|\mathbf{x}_{t}-\mathbf{x}_{t-1}\|_{1}$, where $c_{\mathrm{e}}>0$ converts the energy per unit displacement of an actuator into the units of the rate reward. Displacing a second antenna within the stroke already claimed by the first incurs energy but no additional delay. The coefficient $c_{\mathrm{e}}$ is determined by the actuator hardware and is fixed once the array is deployed.

Let $\mathbf{H}(\mathbf{x};\boldsymbol{\xi}^{(t)})\in\mathbb{C}^{M\times K}$ denote the channel matrix from the array at APV $\mathbf{x}$ to the $K$ devices being served, whose $(m,k)$ entry is the multipath field $h_{k}(x_{m};\boldsymbol{\xi}_{k}^{(t)})$ of~\cite{choi2026positioning}. The parameter vector $\boldsymbol{\xi}^{(t)}$ collects, for each device $k$ and path $\ell=1,\dots,L$, the angle of departure $\theta_{k,\ell}^{(t)}$ and the complex gain $\alpha_{k,\ell}^{(t)}$, together with the position $s_{k}^{(t)}$ of the blockage edge that attenuates the line-of-sight path. These parameters evolve in \emph{event-driven} steps rather than being redrawn in each slot, and two kinds of event occur on different time scales. An \emph{environment event}, such as a blockage clearing, a scatterer entering, or an occupant crossing a link, alters the parameters of the affected link while the device itself remains static or nearly so. A \emph{device swap}, which occurs when a scheduler rotates the devices being served, replaces the channel of the affected link with one independent of its predecessor, and takes place on a longer time scale. Each event concerns a single device and the parameters remain fixed between events: each $\boldsymbol{\xi}_{k}^{(t)}$ is piecewise stationary and changes only at the events of device $k$, the statistics of which are specified in Section~\ref{subsec:setup}. The reward map defined below changes at the events of any of the $K$ devices, and the rate of these events governs the variation budget defined in Section~\ref{subsec:prelim}, on which the regret of Theorem~\ref{thm:regret} depends.

The AP serves the $K$ devices with a linear precoder of the form $\mathbf{W}=\mathbf{H}\mathbf{F}(\mathbf{H}^{\mathsf{H}}\mathbf{H})$, normalized to unit-norm columns, where the matrix function $\mathbf{F}(\cdot)$ specifies the precoder: $\mathbf{F}(\boldsymbol{\Psi})=\boldsymbol{\Psi}^{-1}$ for zero-forcing, $(\boldsymbol{\Psi}+\kappa\mathbf{I})^{-1}$ for its regularized form, and $\mathbf{I}$ for maximum-ratio transmission. Let $\gamma_{k}(\mathbf{x};\boldsymbol{\xi}^{(t)})$ denote the resulting SINR of device $k$. The reward is the effective sum rate
\begin{equation}
  g_{t}(\mathbf{x})=\Big(1-\frac{T_{\mathrm{a}}}{T_{\mathrm{s}}}\Big)
  \sum_{k=1}^{K}\log_{2}\big(1+\gamma_{k}(\mathbf{x};
  \boldsymbol{\xi}^{(t)})\big).
  \label{eq:reward}
\end{equation}
The algorithm observes only~\eqref{eq:reward} and is otherwise agnostic to the choice of precoder.

In each slot the AP observes only a noisy reward at its current APV,
\begin{equation}
  r_{t}=g_{t}(\mathbf{x}_{t})+\epsilon_{t},
  \label{eq:observation}
\end{equation}
a single scalar for an $M$-dimensional action, and must select $\mathbf{x}_{t+1}\in\mathcal{A}(\mathbf{x}_{t})$ from the history $\mathcal{F}_{t}=\{(\mathbf{x}_{\tau},r_{\tau})\}_{\tau\le t}$ alone. No CSI is acquired at the AP; the only feedback is the rate achieved by the devices, aggregated into this scalar.

\subsection{Problem Formulation}
\label{subsec:problem}

The AP seeks a policy $\pi$ that maximizes the reward accumulated over the horizon, net of the actuation energy,
\begin{align}
  \max_{\pi}\ &\mathbb{E}\!\left[\sum_{t=1}^{T}\Big(g_{t}(\mathbf{x}_{t})
  -c_{\mathrm{e}}\big\|\mathbf{x}_{t}-\mathbf{x}_{t-1}\big\|_{1}\Big)\right]
  \nonumber\\
  \text{s.t.}\ &\mathbf{x}_{t}=\pi(\mathcal{F}_{t-1})
  \in\mathcal{A}(\mathbf{x}_{t-1}).
  \label{eq:online}
\end{align}
Problem~\eqref{eq:online} cannot be solved directly, since $\{g_{t}\}$ is revealed only causally and then only through the noisy observations in~\eqref{eq:observation}. As a reference, we consider an oracle that knows the entire reward sequence in advance but starts from the same $\mathbf{x}_{0}$, remains subject to the same reachability constraint, and incurs the same actuation energy. Its optimal cumulative reward is
\begin{equation}
  \mathrm{OPT}_{T}\triangleq
  \max_{\substack{\{\mathbf{x}_{t}\}_{t=1}^{T}:\\
  \mathbf{x}_{t}\in\mathcal{A}(\mathbf{x}_{t-1})}}
  \ \sum_{t=1}^{T}\Big(g_{t}(\mathbf{x}_{t})
  -c_{\mathrm{e}}\big\|\mathbf{x}_{t}-\mathbf{x}_{t-1}\big\|_{1}\Big),
  \label{eq:oracle}
\end{equation}
and the performance of $\pi$ is measured by its \emph{dynamic regret}
\begin{equation}
  \mathcal{R}_{T}(\pi)\triangleq\mathrm{OPT}_{T}
  -\sum_{t=1}^{T}\Big(g_{t}(\mathbf{x}_{t})
  -c_{\mathrm{e}}\|\mathbf{x}_{t}-\mathbf{x}_{t-1}\|_{1}\Big),
  \label{eq:regret}
\end{equation}
a random variable owing to the observation noise and any randomization in the policy. The oracle is not implementable, and unlike the single-antenna
case~\cite{choi2026positioning} it becomes intractable to compute beyond small arrays: \eqref{eq:oracle} is a dynamic program over $\mathcal{C}$, whose discretization on the grid of Section~\ref{subsec:setup} grows combinatorially with $M$ and already has about $4\times10^{9}$ states for $M=6$. It therefore serves as a reference for the analysis rather than as a benchmark in the simulations.

Our objective is to design an online policy whose regret is \emph{sublinear} in $T$ with high probability, so that its time-averaged performance approaches that of the oracle as the horizon grows.

\section{Compositional Kernelized Bandit}
\label{sec:algorithm}

We present the \emph{compositional movement-aware UCB} ({\bf CoMoveUCB}), which estimates the reward map over the joint configuration space with a compositional kernel, selects a target configuration under a movement cost, and retains it under a persistence rule.

\subsection{Reward Estimation with a Compositional Kernel}
\label{subsec:kernel}

The history $\mathcal{F}_{t}$ is exploited along two axes: observations at one APV inform the estimate at nearby ones through spatial correlation, while temporal variation renders old observations outdated. The reward is accordingly estimated, as in~\cite{choi2026positioning}, by kernel ridge regression over a sliding window of the $w$ most recent observations. Let $\mathcal{W}_{t}=\{\tau:\max(1,t-w+1)\le \tau\le t\}$ denote the active window. With the kernel vector $\mathbf{k}_{t}(\mathbf{x})=[k(\mathbf{x},\mathbf{x}_{\tau})]_{\tau\in\mathcal{W}_{t}}$, the Gram matrix $\mathbf{K}_{t}=[k(\mathbf{x}_{\tau},\mathbf{x}_{\tau'})]_{\tau,\tau'\in\mathcal{W}_{t}}$, and the observation vector $\mathbf{r}_{t}=[r_{\tau}]_{\tau\in\mathcal{W}_{t}}$, the posterior mean and variance are
\begin{align}
  \mu_{t}(\mathbf{x})
  &=\mathbf{k}_{t}(\mathbf{x})^{\mathsf{T}}
    \big(\mathbf{K}_{t}+\varrho\mathbf{I}\big)^{-1}\mathbf{r}_{t},
  \label{eq:posterior_mean}\\
  \sigma_{t}^{2}(\mathbf{x})
  &=k(\mathbf{x},\mathbf{x})-\mathbf{k}_{t}(\mathbf{x})^{\mathrm{T}}
    \big(\mathbf{K}_{t}+\varrho\mathbf{I}\big)^{-1}
    \mathbf{k}_{t}(\mathbf{x}),
  \label{eq:posterior_var}
\end{align}
where $\varrho>0$ is a ridge regularization parameter. The window length $w$ balances the estimation noise against the staleness of the observations.

The choice of $k$ determines what the estimator can represent and how fast it learns. A kernel $k$ induces a reproducing kernel Hilbert space (RKHS) $\mathcal{H}_{k}$, the class of functions the estimator can represent, and the confidence bounds on which the algorithm relies are valid only for reward maps in that class. A standard kernel on $\mathcal{C}\subset\mathbb{R}^{M}$, such as the squared exponential, factorizes over the coordinates: two configurations are similar only if all $M$ antenna positions are close, so an observation informs the estimate only in its immediate neighborhood, and the maximum information gain, which governs both the confidence radius and the regret, grows exponentially with $M$~\cite{srinivas2010}.

The structure of the reward suggests a better choice. For precoders of the form in Section~\ref{subsec:sysmodel}, the effective channels $\mathbf{H}^{\mathsf{H}}\mathbf{W}$ and the column norms of $\mathbf{W}$ depend on the channel only through its Gram matrix, and so does the reward: $g_{t}(\mathbf{x})=\phi_{t}(\boldsymbol{\Psi}_{t}(\mathbf{x}))$
for some function $\phi_{t}$, with
\begin{equation}
  \boldsymbol{\Psi}_{t}(\mathbf{x})=\mathbf{H}^{\mathsf{H}}\mathbf{H}
  =\sum_{m=1}^{M}\mathbf{h}(x_{m})^{*}\,\mathbf{h}(x_{m})^{\mathsf{T}},
  \label{eq:gram}
\end{equation}
where $\mathbf{h}(x)=[h_{1}(x),\dots,h_{K}(x)]^{\mathsf{T}}$ collects the channels of the $K$ devices at position $x$ in slot $t$. The Gram matrix is a sum of per-antenna terms, and expanding $\phi_{t}$ around the mean of
this sum orders the reward by interaction: first-order terms are unary, each depending on a single antenna, and second-order terms add pairwise interactions, each a product of an entry contributed by one antenna and an entry contributed by another. Since each antenna alters only one of the $M$ terms, the higher-order terms tend to scale with powers of the ratio between the contribution of one antenna and the smallest eigenvalue of $\boldsymbol{\Psi}_{t}$, and weaken as $\boldsymbol{\Psi}_{t}$ becomes well conditioned, as is typical when $K/M$ is small. The additive structure holds for any channel model, while the conditioning depends on the channel; together they suggest that a second-order truncation improves as $K/M$ decreases and grows looser as $K/M$ approaches unity. We therefore model the reward with a compositional kernel that retains interactions up to the second order,
\begin{align}
  &k(\mathbf{x},\mathbf{x}')=\frac{\sigma_{f}^{2}}{M+\binom{M}{2}}
  \Bigg[\sum_{m=1}^{M}k_{\mathrm{u}}(x_{m},x_{m}')\nonumber\\
  &+\sum_{i<j}k_{\mathrm{u}}(x_{i},x_{i}')\,k_{\mathrm{u}}(x_{j},x_{j}')\,
  e^{-\frac{((x_{i}-x_{j})-(x_{i}'-x_{j}'))^{2}}{2\varsigma_{\mathrm{p}}^{2}}}
  \Bigg],
  \label{eq:kernel}
\end{align}
where $k_{\mathrm{u}}$ is the physics-informed kernel of~\cite{choi2026positioning} with its weights scaled to unit total
variance, and $\varsigma_{\mathrm{p}}>0$ is the length scale in the \emph{separation} of the pair. The entries of $\mathbf{h}(x)^{*}\mathbf{h}(x)^{\mathrm{T}}$, including the single-antenna power profiles $|h_{k}(x)|^{2}$ on the diagonal, oscillate at differences of directional cosines, the kind of spatial variation that $k_{\mathrm{u}}$ was designed to represent, and $k_{\mathrm{u}}$ therefore serves as the unary component. In each pairwise term, the product of the two unary kernels alone has as its RKHS the sums of products of single-antenna functions, which is the form of the second-order terms; these carry the dependence of the pair on the absolute positions, through the blockage attenuation, the angular spread of each path, and the cross terms between distinct paths. The separation factor weights these products by the separation of the pair, and through the constant component of $k_{\mathrm{u}}$ it also represents the part that depends on the pair only through $x_{i}-x_{j}$, which arises from terms pairing the same two paths at both antennas. Each summand is positive semidefinite, and so is~\eqref{eq:kernel}. The prefactor normalizes the kernel so that its output scale $\sup_{\mathbf{x}\in\mathcal{C}}k(\mathbf{x},\mathbf{x})=\sigma_{f}^{2}$ does not grow with $M$; this is a scaling convention, which rescales the RKHS norm, and hence the bound $B$ of Assumption~\ref{as:reward}, without changing the RKHS as a set of functions.

Under~\eqref{eq:kernel}, two configurations are similar whenever they share the position of an antenna or the geometry of a pair, so each observation informs the estimate across all configurations sharing that structure. Unlike the additive kernels of high-dimensional Bayesian optimization~\cite{kandasamy2015high,rolland2018high}, which partition or group the coordinates according to a structure learned from data, \eqref{eq:kernel} is motivated by the additive structure of the Gram matrix~\eqref{eq:gram}, couples every pair of antennas, and builds each component from the physics-informed $k_{\mathrm{u}}$. Its maximum information gain is polynomial rather than exponential in $M$ (Lemma~\ref{lem:infogain}). A numerical functional analysis of variance (ANOVA) of the sum rate shows that the share of interactions among three
or more antennas increases with $K/M$, as the conditioning argument above suggests; the algorithm does not require the approximation to be exact and may be run at any loading. Maximizing the acquisition function exactly remains combinatorial for any kernel, and is addressed in Section~\ref{subsec:ca}.

\begin{figure}[t]
\centering
\includegraphics[width=0.95\columnwidth]{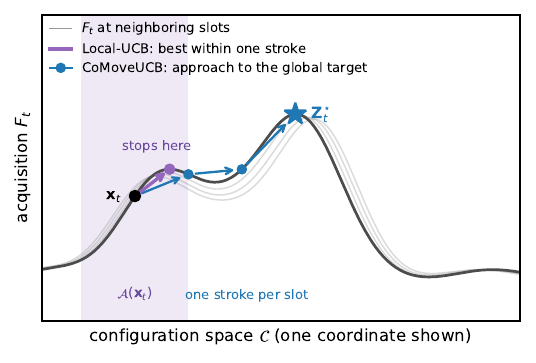}
\caption{Target selection and approach, shown on one coordinate of the acquisition function~\eqref{eq:acq}; the faint curves are $F_{t}$ at neighboring slots. Local-UCB, which maximizes the same $F_{t}$ over the reachable set $\mathcal{A}(\mathbf{x}_{t})$ alone, settles on the peak the array already stands on. CoMoveUCB maximizes over the whole of $\mathcal{C}$ and approaches the target $\mathbf{Z}_{t}^{\star}$ one stroke per slot.}
\label{fig:selection}
\end{figure}

\subsection{Movement-Aware Selection over the Joint Configuration}
\label{subsec:selection}

A displacement incurs three costs, all proportional to distance: the actuation energy of~\eqref{eq:online}; the reward lost while in transit, since the array covers at most $\rho_{\max}$ per slot and transmits from intermediate APVs on the way; and the staleness of the estimate, as $g_{t}$ drifts over those same slots. Only the first is an energy expenditure and only it appears explicitly in~\eqref{eq:regret}, the other two entering through $g_{t}(\mathbf{x}_{t})$, since the array collects the reward of the APV it occupies rather than that of its target. Each unit of travel occupies $1/\rho_{\max}$ slots, over each of which the reward forgone and the drift of the estimate are of order $\sigma_{f}$, so both transit-induced costs scale as $\sigma_{f}/\rho_{\max}$ per unit distance. We price them relative to the energy through a single weight $\eta>0$,
\begin{equation}
  \tilde{c}\;\triangleq\;\frac{1}{w}\Big(c_{\mathrm{e}}
  +\frac{\eta\,\sigma_{f}}{\rho_{\max}}\Big),
  \label{eq:price}
\end{equation}
where $\sigma_{f}^{2}$ is the output scale of~\eqref{eq:kernel}, and the division by the window length $w$ amortizes the one-time cost of a displacement over the $w$ slots during which the resulting observation remains in the window.

The acquisition function is
\begin{equation}
  F_{t}(\mathbf{Z})\triangleq\mu_{t}(\mathbf{Z})
  +\beta\sigma_{t}(\mathbf{Z})
  -\tilde{c}\,\|\mathbf{Z}-\mathbf{x}_{t}\|_{1},
  \label{eq:acq}
\end{equation}
where the confidence parameter $\beta>0$ is chosen so that $\mu_{t}+\beta\sigma_{t}$ upper-bounds $g_{t}$ with high probability. The target is the joint maximizer
\begin{equation}
  \mathbf{Z}_{t}^{\star}=\arg\max_{\mathbf{Z}\in\mathcal{C}}F_{t}(\mathbf{Z}),
  \label{eq:target}
\end{equation}
retained or replaced by the persistence rule
\begin{equation}
  \mathbf{Z}_{t}=
  \begin{cases}
    \mathbf{Z}_{t-1}, & F_{t}(\mathbf{Z}_{t-1})\ge
      F_{t}(\mathbf{Z}_{t}^{\star})-\Delta,\\[4pt]
    \mathbf{Z}_{t}^{\star}, & \text{otherwise},
  \end{cases}
  \label{eq:hysteresis}
\end{equation}
with $\mathbf{Z}_{0}\triangleq\mathbf{x}_{0}$ and a margin $\Delta>0$. Optimism, that is, selection by the upper confidence bound, does not by itself bound the movement, because the estimator forgets: the uncertainty at an APV the array has left grows back once its observations leave the window, so an acquisition function maximized afresh at every slot may redirect the array before it reaches any target. Holding the target unless some configuration is better by $\Delta$ suppresses such redirection, at a loss of at most $\Delta$ per slot.

The target is approached as far as the actuators permit,
\begin{equation}
  x_{t+1,m}=x_{t,m}+\operatorname{sgn}(z_{t,m}-x_{t,m})
  \min\big(\rho_{\max},|z_{t,m}-x_{t,m}|\big),
  \label{eq:approach}
\end{equation}
for all $m$, where $z_{t,m}$ is the $m$th component of $\mathbf{Z}_{t}$. This update is the projection of $\mathbf{Z}_{t}$ onto the $\ell_{\infty}$ ball of radius $\rho_{\max}$ around $\mathbf{x}_{t}$, which is separable across coordinates. Since clipping each coordinate to an interval of common width preserves the ordering and the spacing of the antennas, the update remains in $\mathcal{C}$ whenever $\mathbf{x}_{t}$ and $\mathbf{Z}_{t}$ do, and is therefore also the projection onto $\mathcal{A}(\mathbf{x}_{t})$. Together, \eqref{eq:target}--\eqref{eq:approach} form the persistent target-and-approach scheme. As in the single-antenna design, separating the target from the step toward it is what allows a distant configuration to be selected at all (Fig.~\ref{fig:selection}), and it is also what the analysis requires: the maximization in~\eqref{eq:target} ranges over the whole of $\mathcal{C}$, so the regret analysis may compare $\mathbf{Z}_{t}$ against an arbitrary fixed APV rather than against the configurations reachable within one slot.

\subsection{Coordinate Ascent for Target Selection}
\label{subsec:ca}

With the kernel~\eqref{eq:kernel} the posterior mean and the movement term are sums of unary and pairwise terms, so already for $\beta=0$ the maximization in~\eqref{eq:target} amounts to maximum a posteriori inference on a complete pairwise graph, which is NP-hard in general~\cite{wainwright2008}. Among the standard approximations---max-product belief propagation, linear-programming relaxations, and coordinate ascent---we adopt the last, for a reason specific to~\eqref{eq:hysteresis}. Initialized at $\mathbf{Z}^{(0)}=\mathbf{Z}_{t-1}$ and with the incumbent retained among the candidates, as specified below, coordinate ascent never decreases $F_{t}$ along the sweep, so its output satisfies $F_{t}(\mathbf{Z}_{t}^{\star})\ge F_{t}(\mathbf{Z}_{t-1})$ by construction, which is precisely the comparison the persistence rule makes. Belief propagation offers no such guarantee on a graph with cycles, a relaxation would have to be solved and rounded afresh in every slot, and neither applies directly to the confidence term, which is not pairwise.

Writing $\mathbf{Z}^{(r)}$ for the iterate after $r$ rounds, each round updates the coordinates in turn,
\begin{equation}
  z^{(r)}_{m}=\arg\max_{z\in\mathcal{Z}^{(r)}_{m}}\ F_{t}\big(
  z^{(r)}_{1},\dots,z^{(r)}_{m-1},\,z,\,
  z^{(r-1)}_{m+1},\dots,z^{(r-1)}_{M}\big),
  \label{eq:coord-ascent}
\end{equation}
for $m=1,\dots,M$, where the candidate set $\mathcal{Z}^{(r)}_{m}$ collects the points of the uniform grid $\mathcal{X}_{g}=\{0,\Delta_{g},\dots,D\}$ of $G$ points, with $\Delta_{g}\ll\lambda/2$, that satisfy the spacing constraint of~\eqref{eq:array-set} given the current values of the other coordinates, together with the incumbent $z^{(r-1)}_{m}$, so that $\mathcal{Z}^{(r)}_{m}$ is never empty. After $R$ rounds we set $\mathbf{Z}^{\star}_{t}\gets\mathbf{Z}^{(R)}$. Each update in~\eqref{eq:coord-ascent} is exact because $F_{t}$ is evaluated on a finite set, and since that set contains the incumbent, $F_{t}$ is nondecreasing along the sweep. With the fixed number $R$ of rounds used here, $\mathbf{Z}^{(R)}$ is no worse than $\mathbf{Z}^{(0)}$ but is in general neither a coordinatewise maximum nor the global maximizer of~\eqref{eq:target}.

In~\eqref{eq:coord-ascent} only the kernel terms involving $z_{m}$ change, namely one unary term and the $M-1$ pairwise terms containing $z_{m}$, so the posterior mean at a candidate is updated in $O(Mw)$, while the posterior variance requires a triangular solve against the Cholesky factor of the slot, in $O(w^{2})$. A single update therefore costs $O\big(G(Mw+w^{2})\big)$ and the $R$ rounds $O\big(RMGw(M+w)\big)$, which is polynomial in $M$, against the $O(G^{M})$ of an exhaustive maximization. The analysis of Section~\ref{sec:analysis} assumes that~\eqref{eq:target} is solved exactly, and Remark~\ref{rem:exact} reports the gap between $\mathbf{Z}^{(R)}$ and the exhaustive maximizer for small $M$, where the latter is computable.

\begin{algorithm}[t]
\caption{CoMoveUCB}
\label{alg:comoveucb}
\begin{algorithmic}[1]
\Require actuation parameters $\rho_{\max}$, $c_{\mathrm{e}}$, $d$;
  kernel $k_{\mathrm{u}}$, $\varsigma_{\mathrm{p}}$, $\sigma_{f}$;
  window $w$, weight $\eta$, margin $\Delta$, confidence $\beta$,
  ridge $\varrho$, rounds $R$, grid $\mathcal{X}_{g}$;
  initial APV $\mathbf{x}_{0}\in\mathcal{C}$ and initial
  observations $\mathcal{W}_{0}$
\State $\tilde{c}\gets\big(c_{\mathrm{e}}+\eta\sigma_{f}/\rho_{\max}\big)/w$
  \Comment{\eqref{eq:price}}
\State $\mathcal{W}\gets\mathcal{W}_{0}$,\quad
  $\mathbf{Z}_{0}\gets\mathbf{x}_{0}$,\quad
  $\mathbf{x}_{1}\gets\mathbf{x}_{0}$
\For{$t=1,2,\dots,T$}
  \State Observe $r_{t}=g_{t}(\mathbf{x}_{t})+\epsilon_{t}$;
    append $(\mathbf{x}_{t},r_{t})$ to $\mathcal{W}$ and retain the
    $w$ most recent entries
  \State Update $\mu_{t},\sigma_{t}$ from $\mathcal{W}$
    \Comment{\eqref{eq:posterior_mean}--\eqref{eq:posterior_var}}
  \State $\mathbf{Z}^{(0)}\gets\mathbf{Z}_{t-1}$
  \For{$r=1,\dots,R$}
    \For{$m=1,\dots,M$}
      \State $z^{(r)}_{m}\gets\arg\max_{z\in\mathcal{Z}^{(r)}_{m}}
        F_{t}(\cdot)$
        \Comment{\eqref{eq:coord-ascent}}
    \EndFor
  \EndFor
  \State $\mathbf{Z}_{t}^{\star}\gets\mathbf{Z}^{(R)}$
  \If{$F_{t}(\mathbf{Z}_{t-1})<F_{t}(\mathbf{Z}_{t}^{\star})-\Delta$}
    \State $\mathbf{Z}_{t}\gets\mathbf{Z}_{t}^{\star}$
  \Else
    \State $\mathbf{Z}_{t}\gets\mathbf{Z}_{t-1}$
  \EndIf
  \Comment{\eqref{eq:hysteresis}}
  \State $\mathbf{x}_{t+1}\gets$ coordinatewise step from
    $\mathbf{x}_{t}$ toward $\mathbf{Z}_{t}$
    \Comment{\eqref{eq:approach}}
\EndFor
\end{algorithmic}
\end{algorithm}

\subsection{Complexity}
\label{subsec:complexity}

The complete procedure is summarized in Algorithm~\ref{alg:comoveucb}. Each slot requires the Cholesky factor of the regularized Gram matrix, which costs $O(w^{3})$ from scratch but only $O(w^{2})$ when updated and downdated by one observation as the window slides, together with $O(M^{2}w)$ operations to evaluate the kernel between the new observation and the window. The $R$ rounds of coordinate ascent~\eqref{eq:coord-ascent} cost $O\big(RMGw(M+w)\big)$, as derived in Section~\ref{subsec:ca}, the factor $M$ inside the parentheses arising from the pairwise terms of~\eqref{eq:kernel} and the factor $w$ from the posterior variance, and dominate the per-slot cost. The dependence on $M$ is thus polynomial, in contrast to the $O(G^{M})$ of an exhaustive joint maximization.

For $M=1$ the pairwise terms of~\eqref{eq:kernel} vanish and the kernel becomes $\sigma_{f}^{2}k_{\mathrm{u}}$, the $\ell_{\infty}$ and $\ell_{1}$ norms coincide with the absolute value, and the coordinate ascent turns into an exhaustive search over $\mathcal{X}_{g}$. CoMoveUCB then reduces to the single-antenna algorithm of~\cite{choi2026positioning}.

\section{Regret Analysis}
\label{sec:analysis}

We establish a sublinear bound on the dynamic regret~\eqref{eq:regret}. The proof follows the decomposition of~\cite{choi2026positioning}, and its steps carry over to the joint configuration space, with $\mathcal{D}$ replaced by $\mathcal{C}$, except in three places. The information gain of the compositional kernel~\eqref{eq:kernel} must be bounded anew, and Section~\ref{subsec:infogain} shows it to be polynomial in $M$. The travel of $M$ coordinates moving under a common target, with delay and energy governed by different norms, is bounded in Section~\ref{subsec:travel}. Finally, the exact maximization of the acquisition function, which~\cite{choi2026positioning} obtains by enumeration, becomes an assumption here, and its relaxation is discussed in Remark~\ref{rem:exact}.

\subsection{Preliminaries}
\label{subsec:prelim}

The analysis rests on three assumptions. The first two are standard in kernelized bandits~\cite{srinivas2010,chowdhury2017}, and the third, which those analyses make implicitly, is stated explicitly here because the maximization is combinatorial; it separates the statistical problem from the combinatorial one.

\begin{assumption}[Reward]
\label{as:reward}
For every $t$, $g_{t}$ belongs to the RKHS $\mathcal{H}_{k}$ of the kernel~\eqref{eq:kernel} with $\|g_{t}\|_{\mathcal{H}_{k}}\le B$ and
$0\le g_{t}(\mathbf{x})\le g_{\max}$ for all $\mathbf{x}\in\mathcal{C}$.
\end{assumption}

\begin{assumption}[Noise]
\label{as:noise}
The observation noise $\epsilon_{t}$ in~\eqref{eq:observation} is conditionally zero-mean and $\sigma_{\epsilon}$-sub-Gaussian given the history.
\end{assumption}

\begin{assumption}[Exact target selection]
\label{as:exact}
The maximization~\eqref{eq:target} is solved exactly at every slot.
\end{assumption}

Assumption~\ref{as:exact} holds by construction whenever $\mathcal{X}_{g}^{M}$ is small enough to enumerate, as for a single antenna or a small array, and is approximated by the coordinate ascent of Section~\ref{subsec:ca} otherwise; Remark~\ref{rem:exact} quantifies the resulting gap where enumeration is feasible. The nonstationarity is measured by the variation budget
\begin{equation}
  V_{T}\triangleq\sum_{t=2}^{T}\big\|g_{t}-g_{t-1}\big\|_{\infty},
  \label{eq:variation}
\end{equation}
and the statistical difficulty by the maximum information gain
\begin{equation}
  \gamma_{n}(k)\triangleq\max_{A}
  \tfrac{1}{2}\log\det\!\big(\mathbf{I}
  +\varrho^{-1}\mathbf{K}_{A}\big),
  \label{eq:infogain}
\end{equation}
where the maximum is over multisets $A$ of $n$ points of $\mathcal{C}$ and $\mathbf{K}_{A}$ is the Gram matrix of $k$ on $A$. We write $\gamma_{n}$ for $\gamma_{n}(k)$ with the kernel~\eqref{eq:kernel}.

\subsection{Information Gain of the Compositional Kernel}
\label{subsec:infogain}

The maximum information gain~\eqref{eq:infogain} controls the confidence radius and the estimation term of the regret, and its growth in $M$ decides whether Theorem~\ref{thm:regret} remains useful for an array. The kernel~\eqref{eq:kernel} is a sum of $M+\binom{M}{2}$ components, each depending on at most two coordinates. Lemma~\ref{lem:subadd} bounds the gain of a sum of kernels by the sum of the individual gains, and Lemma~\ref{lem:infogain} applies it to~\eqref{eq:kernel}, yielding $O\big(M^{2}(\log n)^{3}\big)$ in place of the $O\big((\log n)^{M+1}\big)$ of a squared exponential kernel on the joint space~\cite{srinivas2010}, at the cost of the model error discussed in Section~\ref{subsec:kernel}.

\vspace{0.1cm}
\begin{lemma}[Subadditivity of the information gain]
\label{lem:subadd}
Let $k=k^{(1)}+k^{(2)}$ with $k^{(1)},k^{(2)}$ positive semidefinite. Then $\gamma_{n}(k)\le\gamma_{n}(k^{(1)})+\gamma_{n}(k^{(2)})$ for every $n$.
\end{lemma}
\begin{IEEEproof}
Fix a multiset $A$ of $n$ points and write $\mathbf{A}=\varrho^{-1}\mathbf{K}^{(1)}_{A}$ and $\mathbf{B}=\varrho^{-1}\mathbf{K}^{(2)}_{A}$, both positive semidefinite. Since $\mathbf{I}+\mathbf{A}$ is positive definite,
\begin{equation*}
  \det\!\big(\mathbf{I}+\mathbf{A}+\mathbf{B}\big)
  =\det\!\big(\mathbf{I}+\mathbf{A}\big)\,
   \det\!\big(\mathbf{I}+\mathbf{C}\big),
   \qquad
   \mathbf{C}\triangleq\mathbf{S}\mathbf{B}\mathbf{S},
\end{equation*}
with $\mathbf{S}\triangleq(\mathbf{I}+\mathbf{A})^{-1/2}$. The matrix $\mathbf{C}$ is congruent to $\mathbf{B}$, so by Ostrowski's theorem~\cite{hornjohnson} its eigenvalues are $\lambda_{i}(\mathbf{C})=\theta_{i}\lambda_{i}(\mathbf{B})$ with each $\theta_{i}$ between the smallest and the largest eigenvalue of $\mathbf{S}\mathbf{S}^{\mathsf{T}}=(\mathbf{I}+\mathbf{A})^{-1}$. Since $\mathbf{A}\succeq0$ gives $\mathbf{I}+\mathbf{A}\succeq\mathbf{I}$ and hence $\lambda_{\max}\big((\mathbf{I}+\mathbf{A})^{-1}\big)\le1$, we have $\theta_{i}\le1$ and $\det(\mathbf{I}+\mathbf{C})=\prod_{i}(1+\theta_{i}\lambda_{i}(\mathbf{B}))
\le\prod_{i}(1+\lambda_{i}(\mathbf{B}))=\det(\mathbf{I}+\mathbf{B})$. Thus $\det(\mathbf{I}+\mathbf{A}+\mathbf{B})\le \det(\mathbf{I}+\mathbf{A})\det(\mathbf{I}+\mathbf{B})$ for every $A$.
Taking $\tfrac{1}{2}\log$ and bounding each factor by its maximum over $A$, which can only increase the right-hand side, gives the claim.
\end{IEEEproof}

\vspace{0.1cm}
\begin{lemma}[Information gain of the compositional kernel]
\label{lem:infogain}
Let $\gamma^{\mathrm{u}}_{n}$ and $\gamma^{\mathrm{p}}_{n}$ denote the maximum information gain of a unary and a pairwise summand of~\eqref{eq:kernel}, prefactor included, on $\mathcal{D}$ and $\mathcal{D}^{2}$ respectively, the maximum in~\eqref{eq:infogain} being taken over multisets of $n$ points. Then the kernel~\eqref{eq:kernel} satisfies
\begin{equation}
  \gamma_{n}\;\le\;M\,\gamma^{\mathrm{u}}_{n} +\tbinom{M}{2}\,\gamma^{\mathrm{p}}_{n} \;=\;O\big(M^{2}(\log n)^{3}\big).
  \label{eq:infogain-bound}
\end{equation}
\end{lemma}
\begin{IEEEproof}
The kernel~\eqref{eq:kernel} is the sum of $M$ unary and $\binom{M}{2}$ pairwise summands, each positive semidefinite, and Lemma~\ref{lem:subadd} applied inductively bounds $\gamma_{n}$ by the sum of their individual information gains. The $m$th unary summand, $\nu\,k_{\mathrm{u}}(x_{m},x_{m}')$ with $\nu \triangleq\sigma_{f}^{2}/\big(M+\tbinom{M}{2}\big)$, depends on $\mathbf{x}$ only through $x_{m}$, so its Gram matrix on any multiset $A=\{\mathbf{x}_{1},\dots,\mathbf{x}_{n}\}$ drawn from $\mathcal{C}$ coincides with that of $\nu \,k_{\mathrm{u}}$ on the projected multiset $\{x_{1,m},\dots,x_{n,m}\}$ in $\mathcal{D}$, whose entries may repeat when two configurations share that coordinate. Maximizing over multisets from $\mathcal{C}$ is therefore no larger than maximizing over multisets of $n$ points in $\mathcal{D}$, which is $\gamma^{\mathrm{u}}_{n}$. The same argument applied to the pair $(x_{i},x_{j})$ bounds each pairwise summand by $\gamma^{\mathrm{p}}_{n}$, the projection now landing in $\mathcal{D}^{2}$.

For the rates we use two facts. First, if two stationary kernels $k'$ and $k''$ have spectral densities and that of $k'$ is dominated pointwise by that of $c\,k''$ for a constant $c$, then $c\,k''-k'$ is positive semidefinite by Bochner's theorem~\cite{rasmussen2006}, so that $\mathbf{K}'_{A}\preceq c\,\mathbf{K}''_{A}$ for the Gram matrices on every multiset $A$, and $\gamma_{n}(k')\le\gamma_{n}(c\,k'')$ by the monotonicity of the log-determinant~\cite{hornjohnson}. Second, a squared exponential kernel on a compact subset of $\mathbb{R}^{d}$ has maximum information gain $O\big((\log n)^{d+1}\big)$~\cite{srinivas2010}. This holds for any output scale, since a constant factor on a kernel enters~\eqref{eq:infogain} only through the effective ridge and affects the constants but not the rates.

The kernel $k_{\mathrm{u}}$ is a constant plus $Q+1$ components with Gaussian spectral densities, the $Q$ fringes each contributing a symmetric pair of lobes~\cite{choi2026positioning}. The constant has a rank-one Gram matrix and contributes $O(\log n)$, while each component is dominated by a squared exponential of shorter length scale, so Lemma~\ref{lem:subadd} gives $\gamma^{\mathrm{u}}_{n}=O\big((\log n)^{2}\big)$. Expanding the pairwise summand with this decomposition of $k_{\mathrm{u}}$ yields four kinds of term: the product of the two constants with the separation factor depends on $x_{i}-x_{j}$ alone, as noted in Section~\ref{subsec:kernel}, while the remaining products involve one or both non-constant parts and depend on two linearly independent combinations of $(x_{i},x_{j})$. Each is dominated in the same way by a squared exponential in one or two dimensions, so Lemma~\ref{lem:subadd} gives $\gamma^{\mathrm{p}}_{n}=O\big((\log n)^{3}\big)$. Both rates are unaffected by the prefactor $\nu$ of the summands, which by the second fact enters the constants only. Substituting and using $\binom{M}{2}\le M^{2}$ gives~\eqref{eq:infogain-bound}.
\end{IEEEproof}

\subsection{Bounding the Travel}
\label{subsec:travel}

Let $N_{\mathrm{sw}}$ denote the number of slots at which the persistence rule~\eqref{eq:hysteresis} replaces the target, and call slot $t$ a \emph{transit} slot if $\mathbf{x}_{t+1}\neq\mathbf{Z}_{t}$ and a \emph{dwell} slot otherwise. Optimism controls the reward forgone at each slot but says nothing about the distance the array covers, and the following lemma supplies that missing bound, charging the travel to the number of times the target changes.

\vspace{0.1cm}
\begin{lemma}[Travel]
\label{lem:travel}
Under~\eqref{eq:hysteresis}--\eqref{eq:approach},
\begin{align}
  \sum_{t=1}^{T}\big\|\mathbf{x}_{t}-\mathbf{x}_{t-1}\big\|_{1}
  \le MD\big(N_{\mathrm{sw}}+1\big),\nonumber\\
  T_{\mathrm{trans}}\le\frac{D\big(N_{\mathrm{sw}}+1\big)}{\rho_{\max}},
  \label{eq:travel}
\end{align}
where $T_{\mathrm{trans}}$ is the number of transit slots.
\end{lemma}
\begin{IEEEproof}
Partition the horizon at the switch slots into at most $N_{\mathrm{sw}}+1$ maximal journeys on which the target is fixed. Consider one such journey, with target $\mathbf{Z}$ and initial configuration $\mathbf{x}$. By~\eqref{eq:approach} each coordinate moves monotonically toward its target component and halts on arrival, so coordinate $m$ covers a total distance $|z_{m}-x_{m}|\le D$, and the $\ell_{1}$ travel of the journey is at most $MD$. Summing over the journeys gives the first bound.

For the second, write $d_{t}\triangleq\|\mathbf{Z}-\mathbf{x}_{t}\|_{\infty}$ for the remaining $\ell_{\infty}$ distance. If $d_{t}\le\rho_{\max}$, every coordinate reaches its target component in one slot by~\eqref{eq:approach}, and slot $t$ is a dwell slot. On a transit slot $d_{t}>\rho_{\max}$, and since each coordinate closes its own gap by $\min(\rho_{\max},|z_{m}-x_{t,m}|)$, the remaining distance satisfies $d_{t+1}=\max_{m}\max\{|z_{m}-x_{t,m}|-\rho_{\max},0\}=d_{t}-\rho_{\max}$. A journey starting at slot $t_{0}$ therefore contains at most $d_{t_{0}}/\rho_{\max}\le D/\rho_{\max}$ transit slots, and summing over the journeys gives the second bound.
\end{IEEEproof}

Under Assumption~\ref{as:exact}, the switch count is controlled exactly as in the single-antenna case~\cite[Lemma~3]{choi2026positioning}, the argument depending only on the margin in~\eqref{eq:hysteresis} and not on the dimension. With
$S_{T}\triangleq\sum_{t=1}^{T-1}\sup_{\mathbf{Z}\in\mathcal{C}}
|F_{t+1}(\mathbf{Z})-F_{t}(\mathbf{Z})|$ measuring how fast the
acquisition function moves,
\begin{equation}
  N_{\mathrm{sw}}\;\le\;1+\frac{2S_{T}}{\Delta}.
  \label{eq:switch}
\end{equation}
The drift $S_{T}$ is itself bounded as in~\cite[Lemma~6]{choi2026positioning}, with $\gamma_{w}$ now given by Lemma~\ref{lem:infogain}; the substitution is carried out in the proof of Theorem~\ref{thm:regret}. Combining~\eqref{eq:travel} with the per-slot transit loss of at most $g_{\max}$ (Assumption~\ref{as:reward}), the actuation energy and the transit loss together are bounded by $C_{0}(N_{\mathrm{sw}}+1)$ with
\begin{equation}
  C_{0}\triangleq D\left(c_{\mathrm{e}}M+\frac{g_{\max}}{\rho_{\max}}\right).
  \label{eq:c0}
\end{equation}

\subsection{Main Result}
\label{subsec:main}

Combining the information-gain bound of Lemma~\ref{lem:infogain} and the travel bound of Lemma~\ref{lem:travel} with the decomposition of~\cite{choi2026positioning} yields the following guarantee.

\vspace{0.1cm}
\begin{theorem}
\label{thm:regret}
Under Assumptions~\ref{as:reward}--\ref{as:exact}, with the window and margin set to $w=\Theta\big(M(T/V_{T})^{1/2}\big)$ and $\Delta=\Theta\big(M^{5/4}(V_{T}/T)^{1/8}\big)$, for any $\delta\in(0,1)$ Algorithm~\ref{alg:comoveucb} attains, with probability at least $1-\delta$,
\begin{equation}
  \mathcal{R}_{T}\;=\;\tilde{O}\big(M^{5/4}\,T^{7/8}V_{T}^{1/8}\big),
  \label{eq:regret-bound}
\end{equation}
where $\tilde{O}$ hides factors polylogarithmic in $T$ and $1/\delta$. Setting $\delta=1/T$, the same order bounds $\mathbb{E}[\mathcal{R}_{T}]$. Both bounds are sublinear in $T$ whenever $V_{T}=o(T)$.
\end{theorem}

\begin{IEEEproof}
The decomposition of~\cite[App.~A]{choi2026positioning} applies with $\mathcal{D}$ replaced by $\mathcal{C}$ and $x_{t}$ by $\mathbf{x}_{t}$. Its optimism step compares the selected target against an arbitrary fixed element of the set over which the acquisition function is maximized, which by~\eqref{eq:target} and Assumption~\ref{as:exact} is $\mathcal{C}$, the set from which the oracle~\eqref{eq:oracle} also chooses. Writing $\mathcal{R}_{T}=\mathcal{R}^{\mathrm{val}}_{T}+\mathcal{R}^{\mathrm{mov}}_{T}
+g_{\max}T_{\mathrm{trans}}$ as there, with $g_{\max}$ in place of $2g_{\max}$ since $0\le g_{t}\le g_{\max}$, three quantities differ from the single-antenna case.

The estimation term carries $\sqrt{\gamma_{w}}$, which is $O\big(M(\log w)^{3/2}\big)$ by Lemma~\ref{lem:infogain} rather than the
one-dimensional $O(\log w)$; the confidence radius $\beta=O(\sqrt{\gamma_{w}})$ inherits the same factor $M$, and so does the constant $a$ of~\cite[Lemma~6]{choi2026positioning}, which is $O(\beta)$ at a fixed output scale. The actuation energy and the transit loss are bounded by $C_{0}(N_{\mathrm{sw}}+1)$ with $C_{0}=\Theta(M)$ from~\eqref{eq:c0}, and $N_{\mathrm{sw}}+1\le2+2S_{T}/\Delta$ by~\eqref{eq:switch}. The maximization~\eqref{eq:target} is carried out on $\mathcal{X}_{g}^{M}$ rather than on $\mathcal{X}_{g}$. Rounding each coordinate of a configuration in $\mathcal{C}$ to the nearest grid point keeps it in $\mathcal{C}$ when $\Delta_{g}$ divides $d$, and changes $F_{t}$ by at most $ML_{U}\Delta_{g}/2$, where $L_{U}$ is the per-coordinate Lipschitz constant of~\cite[App.~A]{choi2026positioning}, itself $O(M)$ through $\beta$ and $\gamma_{w}$. The total grid error is therefore $O(M^{2}\Delta_{g}T)$, which is $O(M^{2})$ under $\Delta_{g}=O(D/T)$ as in~\cite{choi2026positioning} and is dominated.

The margin $\Delta$ trades the movement term $2C_{0}S_{T}/\Delta$ against the loss $\Delta T$ incurred by holding the target, and balancing the two gives $\Delta=\Theta\big(\sqrt{C_{0}S_{T}/T}\big)$ and a movement contribution of $\Theta\big(\sqrt{C_{0}S_{T}T}\big)$. By~\cite[Lemma~6]{choi2026positioning} the drift splits as $S_{T}=O\big(aT\sqrt{\gamma_{w}/w}+\varrho^{-3/2}\sigma_{f}^{3}w^{3/2}V_{T}\big)$, whose first part carries $a\sqrt{\gamma_{w}}=O(M^{2})$ up to logarithmic factors while the second is free of $M$. With $C_{0}=\Theta(M)$, the two parts of $C_{0}S_{T}T$ are of order $M^{3}T^{2}w^{-1/2}$ and $MTw^{3/2}V_{T}$, so that
\begin{equation*}
  \sqrt{C_{0}S_{T}T} =\tilde{O}\big(M^{3/2}\,T\,w^{-1/4} +M^{1/2}T^{1/2}w^{3/4}V_{T}^{1/2}\big).
\end{equation*}
Equating the two terms gives $w=\Theta\big(M(T/V_{T})^{1/2}\big)$, at which the movement contribution is $\tilde{O}\big(M^{5/4}T^{7/8}V_{T}^{1/8}\big)$ and the margin is $\Delta=\Theta\big(M^{5/4}(V_{T}/T)^{1/8}\big)$ as stated. The estimation and staleness terms of~\cite[App.~A]{choi2026positioning} are both $\tilde{O}\big(M^{3/2}T^{3/4}V_{T}^{1/4}\big)$ at this window and are dominated whenever $V_{T}=o(T)$, which gives~\eqref{eq:regret-bound}.

Finally, a single slot contributes at most $g_{\max}+c_{\mathrm{e}}M\rho_{\max}$ to the regret, so the complement of the confidence event, of probability at most $\delta$, contributes at most $\delta T(g_{\max}+c_{\mathrm{e}}M\rho_{\max})$ to $\mathbb{E}[\mathcal{R}_{T}]$, which is $O(M)$ at $\delta=1/T$, while $\ln(1/\delta)=\ln T$ is absorbed into $\tilde{O}$.
\end{IEEEproof}

\begin{remark}[Precoder and regularity]
\label{rem:precoder}
Assumption~\ref{as:reward} constrains the precoder only through the smoothness it induces in $g_{t}$. Regularized zero-forcing with a fixed $\kappa>0$ supports that smoothness uniformly, since the smallest eigenvalue of $\boldsymbol{\Psi}_{t}+\kappa\mathbf{I}$ is at least $\kappa$ for every APV and every channel realization, which keeps $g_{t}$ a smooth function of the APV with derivatives bounded independently of the channel. Pure zero-forcing admits no such deterministic bound, since two devices at similar angles render $\boldsymbol{\Psi}_{t}$ nearly singular wherever the antennas are placed, and the same regularity then holds only with high probability when the array is large relative to the number of devices. The distinction concerns the analysis alone, not the applicability of the algorithm, which observes only the reward.
\end{remark}

\begin{remark}[On Assumption~\ref{as:exact}]\label{rem:exact}
Theorem~\ref{thm:regret} separates the statistical guarantee from the tractability of~\eqref{eq:target}. If the latter is solved only approximately, with an average shortfall $\varepsilon_{\mathrm{opt}}$ from the maximum over the horizon, the optimism step loses that amount per slot on average and the bound acquires an additive $\varepsilon_{\mathrm{opt}}T$; the guarantee degrades gracefully, but is no longer sublinear unless $\varepsilon_{\mathrm{opt}}$ vanishes with $T$. Measured against exhaustive search on a coarsened grid at $M=4$, $K=2$, where enumeration is feasible, the single-start ascent of Section~\ref{subsec:ca} attains $\varepsilon_{\mathrm{opt}}=0.08$, or $0.03\sigma_{f}$; restarting from three initial points halves it at proportionate cost, without a measurable change in the objective at $M=6$, $K=4$.
\end{remark}

\section{Simulation Results}
\label{sec:sim}

We evaluate CoMoveUCB against benchmark policies across loadings and report an ablation study of its key components.

\begin{table}[t]
\centering
\caption{Kernel and algorithm parameters.}
\label{tab:params}
\renewcommand{\arraystretch}{1.1}
\begin{tabular}{@{}lll@{}}
\toprule
& Parameter & Value \\
\midrule
\multirow{5}{*}{Kernel}
 & Fringe components $Q$        & $3$ \\
 & Fringe wavelengths $\{\lambda_q\}$ & $\{4.15,1.81,1.04\}\lambda$ \\
 & Fringe / low-freq.\ damping $\varsigma_q,\varsigma_0$ & $2\lambda$, $4\lambda$ \\
 & Pairwise length scale $\varsigma_{\mathrm{p}}$ & $\lambda$ \\
 & Ridge $\varrho$               & $0.1$ \\
\midrule
\multirow{6}{*}{Algorithm}
 & Grid step $\Delta_g$          & $0.03\lambda$ \\
 & Window $w$                    & $80$ \\
 & Confidence $\beta$            & $0.5$ \\
 & Movement weight $\eta$        & $0.11$ \\
 & Margin $\Delta/\sigma_f$      & $0.022$ \\
 & Ascent rounds $R$             & $1$ \\
\bottomrule
\end{tabular}
\end{table}

\subsection{Setup}
\label{subsec:setup}

The AP is equipped with $M=6$ MAs along a region of length $D=6\lambda$ with a minimum spacing of $d=\lambda/2$, and serves $K=4$ IoT devices. The spacing constraint leaves each antenna an interval of $D-(M-1)d=3.5\lambda$. A slot lasts $T_{\mathrm{s}}=5$~ms, of which $T_{\mathrm{a}}=1$~ms is spent actuating, so an antenna advances at most $\rho_{\max}=0.06\lambda$ per slot and a full reconfiguration of the array takes about $58$ slots. The horizon of $T=6000$ slots spans $30$~s.

Each link has $L=3$ paths, one line-of-sight and two scattered, with powers $(0.6,0.2,0.2)$, arriving from a sector of $[30^\circ,150^\circ]$ and each spread over $8$ sub-rays with an angular spread of $2^\circ$. The line-of-sight path is attenuated by $10$~dB over a blockage spanning half the region. The reward is the sum rate~\eqref{eq:reward} under regularized zero-forcing with $\kappa=0.1$, at a transmit SNR of $9$~dB per device, for which $g_{\max}\approx15$. The movement cost coefficient is set to $c_{\mathrm{e}}=8$ per $\lambda$, about a fifth of the boundary $g_{\max}/(M\rho_{\max})$ beyond which the actuation energy exceeds any achievable gain in rate.

The environment is piecewise stationary: for each device, $\Gamma=20$ environment events and $\Gamma_{\mathrm{d}}=6$ device swaps are placed independently and uniformly over the horizon. An environment event perturbs the angles with standard deviation $\sigma_{J}=0.5^\circ$, refreshes about half of the gain power, and displaces the blockage with standard deviation $\sigma_{J,s}=0.5\lambda$; a device swap redraws the
link's channel independently of its predecessor's. The reward map therefore changes on average every $T/(K(\Gamma+\Gamma_{\mathrm{d}}))\approx 58$ slots, comparable to the time a full reconfiguration takes.

The policies that maintain an estimator are calibrated from $T_{0}=40$ samples of the reward map at $t=0$, which form the initial observations $\mathcal{W}_{0}$ of Algorithm~\ref{alg:comoveucb}: spatial samples of a single frozen channel snapshot at randomly drawn feasible APVs, not measurements collected over a sequence of slots, so no reachability constraint applies to them. This models a site survey obtained before deployment, and consumes none of the $T$ evaluated slots. The output scale $\sigma_{f}^{2}$ is calibrated by assigning each of the $M+\binom{M}{2}$ summands of~\eqref{eq:kernel} the empirical variance of the samples. The evaluated trajectory, subject to~\eqref{eq:reachable} throughout, begins at an equally spaced placement over the central $60\%$ of the region. These settings, with $M=6$, $K=4$, $T=6000$ and $c_{\mathrm{e}}=8$, are common to all results of this section, $K$ being varied only in Table~\ref{tab:loading}. All results are averaged over $48$ independent realizations of the channel and the noise, and margins and error bars denote the difference to the reference policy stated in each figure or table, paired within each realization, with $\pm$ one standard error.

The remaining parameters, of the kernel and the algorithm, are listed in Table~\ref{tab:params}; they were selected in preliminary runs by sweeping each with the others held fixed. The window is set to $w=80$ for three reasons: the reward map changes about every $58$ slots, so an observation older than a few such intervals describes a different map; a longer window weakens the movement price, since $\tilde{c}$ of~\eqref{eq:price} scales as $1/w$; and it raises the cost of the coordinate ascent, which evaluates the posterior $O(RMG)$ times, each at a cost of $O(Mw+w^{2})$.

\subsection{Benchmarks}
\label{subsec:benchmarks}

CoMoveUCB is compared against the following benchmarks, all charged the same actuation energy $c_{\mathrm{e}}\sum_{t}\|\mathbf{x}_{t}-\mathbf{x}_{t-1}\|_{1}$, subject to the same reachability constraint~\eqref{eq:reachable}, and, where they maximize an acquisition over configurations, using the same coordinate ascent.

\vspace{2pt}
\noindent\textbf{Fixed} holds the initial APV throughout. It incurs no actuation energy and requires no estimation, so the margin over it measures the value of repositioning itself.

\vspace{2pt}
\noindent\textbf{Random} takes a random feasible step every slot, with no estimator. Together with Fixed it bounds from below what an informed policy must beat, since repositioning without information incurs actuation energy that the rate gain need not offset.

\vspace{2pt}
\noindent\textbf{Sweep-and-hold} measures before it optimizes. The array translates rigidly, its one degree of freedom being the position of its left edge. Its first cycle measures $T_{0}$ translations during the calibration and costs no evaluated slots; every subsequent cycle, of period $P$, sweeps at full stroke for $T_{0}$ slots, reflecting at the region's boundaries and recording the reward at each position, then holds at the best recorded position until the cycle restarts. These later sweeps are charged in full. The period is selected from $\{500,1000,2000\}$, and $P=2000$ is used throughout.

\vspace{2pt}
\noindent\textbf{$\epsilon$-greedy} shares CoMoveUCB's estimator and moves to the maximizer of the posterior mean within reach, or with exploration probability $0.1$ to a random reachable APV. Its exploration is undirected, unlike the uncertainty-driven exploration of~\eqref{eq:acq}.

\vspace{2pt}
\noindent\textbf{GP-UCB}~\cite{srinivas2010,chowdhury2017} shares CoMoveUCB's kernel, window, and confidence parameter, and selects $\mathbf{x}_{t+1}=\arg\max_{\mathbf{Z}\in\mathcal{A}(\mathbf{x}_{t})}[\mu_{t}(\mathbf{Z})+\beta\sigma_{t}(\mathbf{Z})]$. It prices neither the actuation energy nor the reward lost in transit, and, like Local-UCB in Fig.~\ref{fig:selection}, never evaluates a configuration beyond one stroke.

\begin{figure}[t]
\centering
\includegraphics[width=0.95\columnwidth]{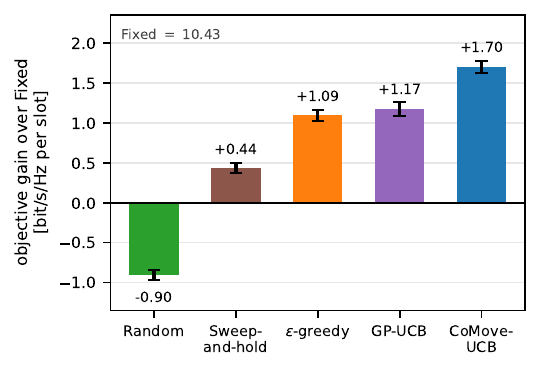}
\caption{Comparison against the benchmarks. Bars show the objective of~\eqref{eq:online} as a gain over the fixed array; error bars denote $\pm$ one standard error.}
\label{fig:main}
\end{figure}

\begin{table}[t]
\centering
\caption{Comparison across loadings.}
\label{tab:loading}
\renewcommand{\arraystretch}{1.15}
\begin{tabular}{@{}ccrrrr@{}}
\toprule
$K$ & $K/M$ & Fixed & Sweep-and-hold & GP-UCB & CoMoveUCB \\
\midrule
$2$ & $0.33$ & $6.38$  & $6.67$  & $6.54$  & $\mathbf{6.97}$  \\
$3$ & $0.50$ & $8.73$  & $9.08$  & $9.38$  & $\mathbf{9.88}$  \\
$4$ & $0.67$ & $10.43$ & $10.87$ & $11.60$ & $\mathbf{12.13}$ \\
\bottomrule
\end{tabular}
\end{table}

\begin{table*}[t]
\centering
\caption{Ablation of CoMoveUCB at $M=6$ and $K=4$. Except for Fixed, all
variants select the target over the whole of $\mathcal{C}$ and price the
movement by~\eqref{eq:price}; each differs from the one it is compared
with in a single component. Margins are paired within each realization
with the variant indicated.}
\label{tab:ablation}
\renewcommand{\arraystretch}{1.15}
\begin{tabular}{@{}llcccrrl@{}}
\toprule
& Variant & Kernel & Optimism & Persistence & Objective & Travel~[$\lambda$] & Margin \\
\midrule
      & Fixed                  & ---           & ---        & ---        & $10.43$ & $0$   & --- \\
(i)   & Posterior-mean target  & Compositional & ---        & ---        & $11.37$ & $70$  & --- \\
(ii)  & Optimistic target      & Compositional & \checkmark & ---        & $12.06$ & $528$ & $+0.70\pm0.08$ vs.\ (i) \\
(iii) & CoMoveUCB              & Compositional & \checkmark & \checkmark & $\mathbf{12.13}$ & $466$ & $+0.07\pm0.06$ vs.\ (ii) \\
\midrule
(iv)  & Unary kernel           & Unary         & \checkmark & \checkmark & $11.41$ & $83$  & $-0.72\pm0.08$ vs.\ (iii) \\
\bottomrule
\end{tabular}
\end{table*}

\subsection{Comparison against the Benchmarks}
\label{subsec:comparison}

Fig.~\ref{fig:main} compares the policies on the objective of~\eqref{eq:online}, as a gain over the fixed array. The policies fall into three groups according to what they can exploit: those that maintain no estimator, those that learn but move only within one stroke, and CoMoveUCB, which learns and selects its target over the whole configuration space. The gain grows from one group to the next.

Without an estimator, repositioning yields no reliable gain. Random falls $0.90$~bit/s/Hz per slot below the fixed array, since undirected movement incurs the actuation energy without a compensating gain in rate. Sweep-and-hold gains $0.44$ by measuring before committing, but achieves only a small fraction of the available gain, for two reasons: its measurements cover only the one-dimensional subset of $\mathcal{C}$ reached by translating the array rigidly, so the spacing between antennas is never adjusted; and the configuration it selects is held for a cycle of $2000$ slots, whereas the reward map changes about every $58$ slots, so the held configuration is outdated for most of each cycle.

The policies that learn continuously, $\epsilon$-greedy and GP-UCB, track the reward map and gain $1.09$ and $1.17$. Both, however, choose each move within one stroke of the current configuration, and therefore tend to settle on a local maximum near where the array already is (cf.\ Fig.~\ref{fig:selection}): a better configuration farther away is never visited, regardless of how promising the estimate makes it. CoMoveUCB selects its target over the whole of $\mathcal{C}$ and approaches it over several slots, accepting the transit loss when the acquisition justifies it, and gains $1.70$. Its margin over GP-UCB, $0.53\pm0.07$, reflects two differences, the range over which the target is sought and the pricing of the movement, as the two policies otherwise share the estimator, kernel, window, and confidence parameter.

Table~\ref{tab:loading} repeats the comparison at $K=2$, $3$, and $4$ with the array size and all hyperparameters held fixed. CoMoveUCB leads at every loading. Its margin over Sweep-and-hold increases with $K/M$, from $0.30\pm0.04$ to $1.26\pm0.09$~bit/s/Hz per slot: separating more devices requires the antennas to move relative to one another, which a rigidly translating array cannot do. Its margin over GP-UCB varies far less, between $0.43\pm0.03$ and $0.53\pm0.07$, so the benefit of looking beyond one stroke is largely insensitive to the loading, whereas that of adjusting the relative positions grows with it.

\subsection{Ablation}
\label{subsec:ablation}

Table~\ref{tab:ablation} builds CoMoveUCB up one component at a time. The baseline (i) selects as its target the configuration with the best posterior mean, net of the movement price, and recomputes it at every slot. Adding the confidence term of~\eqref{eq:acq} yields the optimistic target (ii), and adding the margin of~\eqref{eq:hysteresis} yields CoMoveUCB (iii). Variant (iv) retains both components and replaces the compositional kernel by its unary part.

Optimism accounts for most of the gain, $0.70\pm0.08$~bit/s/Hz per slot from (i) to (ii). Without the confidence term the array moves only toward configurations its estimate already favors, and since an unvisited configuration has an uninformative posterior mean, it rarely leaves its neighborhood, covering $70\lambda$ over the horizon. The uncertainty bonus draws it to unexplored configurations, over $528\lambda$, among which it finds better ones.

Persistence leaves the objective unchanged within the error, $0.07\pm0.06$, while reducing the travel from $528\lambda$ to $466\lambda$. Its role is to secure the guarantee rather than to improve the objective. Optimism certifies the target, not the configuration the array occupies, and the two coincide only once the target is reached. Without the margin, the target may be replaced before the array arrives; if this recurs, the array remains in transit at every slot, the transit loss of Lemma~\ref{lem:travel} grows linearly in $T$, and the bound of~\eqref{eq:regret-bound} becomes vacuous. That case does not arise in these simulations, and the safeguard is obtained at no measurable cost in reward.

Replacing the kernel by its unary part $\sum_{m}k_{\mathrm{u}}(x_{m},x_{m}')$, as in (iv), costs $0.72\pm0.08$~bit/s/Hz per slot. The estimator can then no longer represent how the position of one antenna changes the value of the others, and the array settles on a configuration reached with a fraction of the travel, $83\lambda$ against $466\lambda$. The pairwise terms are what identify a distant configuration as superior to the one currently occupied.

\section{Conclusion}
\label{sec:conclusion}

We addressed the CSI-free positioning of a movable antenna array serving multiple IoT devices under a per-slot displacement limit and an actuation-energy cost. Motivated by the dependence of the multiuser sum rate on the positions through a sum of per-antenna terms, we designed a compositional kernel that retains interactions up to antenna pairs and keeps the maximum information gain polynomial in the number of antennas. On this kernel, CoMoveUCB selects its target over the whole configuration space and retains it under a persistence rule, attaining a dynamic regret sublinear in the horizon with a polynomial dependence on the array size. Simulations showed that it outperforms fixed, measure-then-optimize, and reachability-confined benchmarks at every loading tested, with most of the gain due to optimistic target selection and the pairwise terms of the kernel. Since the AP learns from the achieved rate alone, the approach requires neither pilots nor channel feedback from the devices, which suits low-power IoT deployments. Extending it to planar apertures and validating it on a hardware prototype remain open.

\bibliographystyle{IEEEtran}
\bibliography{IEEEIoTJ_OnlineMA_refs}

\end{document}